\documentclass{article}

\usepackage{fullpage} 
\usepackage{amssymb} 
\usepackage[all,dvips]{xy} 
\usepackage{url}
\usepackage{rotating,stmaryrd}

\usepackage{theorem}
\newtheorem{theo}{Theorem}[section]
\newtheorem{prop}[theo]{Proposition}
\newtheorem{coro}[theo]{Corollary}
\newtheorem{lemm}[theo]{Lemma}
\theorembodyfont{\normalfont}
\newtheorem{defi}[theo]{Definition}
\newtheorem{exam}[theo]{Example}
\newtheorem{rema}[theo]{Remark}

\newenvironment{proof}{\textit{Proof.} }{$\Box$}

\newcommand{\eps}{\varepsilon} 
\newcommand{\cosl}{\!\!\downarrow\!\!} 
\newcommand{\tuple}[1]{\langle#1\rangle} 
\newcommand{\tu}{\langle\,\rangle} 
\newcommand{\curry}[1]{\ulcorner\!{#1}\!\urcorner} 
\newcommand{\kl}[1]{[#1]} 
\newcommand{\PO}{{\scriptstyle [P.O.]}} 
\newcommand{\stimes}{\!\times\!} 
\newcommand{\md}{*}  
\newcommand{\hsp}{\hspace{5mm}}  

\newcommand{\col}[1]{\multicolumn{1}{c}{#1}}

\newcommand{\isoto}{\stackrel{\simeq}{\rightarrow}} 
\newcommand{\iso}{\cong} 
\newcommand{\To}{\Rightarrow} 
\newcommand{\trnat}{\begin{turn}{-20}\ensuremath{\Uparrow}\end{turn}} 
\newcommand{\rpto}{\rightsquigarrow} 

\newcommand{\bZ}{\mathbb{Z}} 
\newcommand{\bL}{\mathbb{L}} 
\newcommand{\bA}{\mathbb{A}} 
\newcommand{\setX}{\mathbb{X}}  
\newcommand{\setY}{\mathbb{Y}}   

\newcommand{\uno}{1}  

\newcommand{\catS}{\mathbf{S}}   
\newcommand{\catT}{\mathbf{T}}   
\newcommand{\eq}{\mathit{eq}} 
\newcommand{\dec}{\mathit{dec}} 
\newcommand{\param}{\mathit{par}} 

\newcommand{\Tt}{\Theta} 
\renewcommand{\tt}{\theta} 
\newcommand{\Ss}{\Sigma} 
\renewcommand{\ss}{\sigma} 
\newcommand{\Set}{\mathit{Set}} 
\newcommand{\Mod}{\mathit{Mod}} 
\newcommand{\sgp}{\mathit{sgp}} 
\newcommand{\mon}{\mathit{mon}} 
\newcommand{\dm}{\mathit{dm}} 
\newcommand{\nat}{\mathit{nat}} 
\newcommand{\st}{\mathit{st}} 
\newcommand{\oper}{\mathit{op}} 

\newcommand{\id}{\mathit{id}} 
\newcommand{\proj}{\mathit{pr}}  
\newcommand{\prd}{\mathit{prd}} 
\newcommand{\unt}{e} 
\newcommand{\dif}{\mathit{dif}} 
\newcommand{\lup}{\mathit{lookup}} 
\newcommand{\upd}{\mathit{update}} 

\newcommand{\ttE}{\mathtt{E}}
\newcommand{\Type}{\mathtt{Type}}
\newcommand{\Term}{\mathtt{Term}}
\newcommand{\Selid}{\mathtt{Selid}}
\newcommand{\Comp}{\mathtt{Comp}}
\newcommand{\bProd}{\mathtt{2\texttt{-}Prod}}
\newcommand{\zProd}{\mathtt{0\texttt{-}Prod}}
\newcommand{\bTuple}{\mathtt{2\texttt{-}Tuple}} 
\newcommand{\zTuple}{\mathtt{0\texttt{-}Tuple}} 
\newcommand{\deco}{\mathtt{.x}} 
\newcommand{\dotp}{\mathtt{.p}} 
\newcommand{\dotg}{\mathtt{.g}} 
\newcommand{\ttc}{\mathtt{c}} 

\title{A parameterization process, functorially} 
\author{
C\' esar Dom\'\i nguez \thanks{
Departamento de Matem\'aticas y Computaci\'on,
Universidad de La Rioja,
Edificio Vives, Luis de Ulloa s/n, E-26004 Logro\~no, La Rioja, Spain,
cesar.dominguez@unirioja.es.}
\and Dominique Duval \thanks{
Laboratoire Jean Kuntzmann, Universit\'e de Grenoble, 
51 rue des math\'ematiques, BP 53, F-38041 Grenoble C\'edex 9, France, 
Dominique.Duval@imag.fr.}
}
\date{August 31., 2009}

\begin{document}

\maketitle

\begin{itemize}
\item[] \textbf{Abstract.}
The parameterization process used in the symbolic computation systems 
Kenzo and EAT is studied here as a general construction in a categorical framework.
This parameterization process starts from a given specification 
and builds a parameterized specification 
by adding a parameter as a new variable to some operations. 
Given a model of the parameterized specification, 
each interpretation of the parameter, called an argument, 
provides a model of the given specification.
Moreover, under some relevant terminality assumption, 
this correspondence between the arguments and the models of the 
given specification is a bijection. 
It is proved in this paper that the parameterization process 
is provided by a functor
and the subsequent parameter passing process by a natural transformation.
Various categorical notions are used, mainly  
adjoint functors, pushouts and lax colimits.
\end{itemize}

\section{Introduction}


Kenzo \cite{Kenzo} and its predecessor EAT \cite{EAT} are software
systems developed by F. Sergeraert. They are devoted to
Symbolic Computation in Algebraic Topology. In particular, they 
carry out calculations of homology groups of complex topological
spaces, namely iterated loop spaces. By means of EAT and Kenzo,
some homology groups that had never been obtained with any other
method, neither theoretical nor automatic, have been computed. In
view of the obtained results, some years ago, the first author 
of this paper began the formal study of the programs, in order to
reach a good understanding on the internal calculation processes
of these software systems.
In particular, our study of the data types used in EAT and
Kenzo \cite{LPR03,DLR07,DRS06} 
shows that there are two different layers of data structures in the systems.
In the first layer, one finds the usual abstract data types,
like the type of integers. 
In the second layer, one deals with algebraic structures,
like the structure of groups, 
which are implemented thanks to the abstract data types belonging to the first layer. 
In addition, we realized that in a system such as EAT, 
we do not simply implement one group, 
but more generally \emph{parameterized} families of groups. 
In \cite{LPR03} an operation is defined, 
which is called the \emph{imp} construction
because of its role in the implementation process in the system EAT.
Starting from a specification $\Ss$
in which some operations are labelled as ``pure'' \cite{DRS06}, 
the \emph{imp} construction builds a new specification $\Ss_A$ 
with a distinguished sort $A$
which is added to the domain of each non-pure operation.
It follows that each implementation of $\Ss_A$ 
defines a family of implementations of $\Ss$ 
depending on the choice of a value in the interpretation of $A$.
Besides, working with the \emph{imp} construction in \cite{LPR03} we were able
to prove that the implementations of EAT algebraic structures are as general as
possible, in the sense that they are ingredients of terminal objects in
certain categories of models;
this result is called the \emph{exact parameterization property}.
Later on, led by this characterization of EAT algebraic structures, 
in \cite{LPR03} we reinterpreted our results in terms of object-oriented technologies
like hidden algebras \cite{GM00} or coalgebras \cite{Ru00}.


This paper deals with generalization by parameterization in the sense of Kenzo and EAT,
so that our \emph{parameters} are symbolic constants of a given type, 
that will be replaced by \emph{arguments} which are elements in a given set.
The notion of parameterization in programming and specification languages 
bears several meanings, where the parameter may be a type or a specification.
For instance, in object-oriented programming, 
parametric polymorphism is called generic programming, 
in C++ it is characterized by the use of template parameters 
to represent abstract data types. 
On the other hand, in algebraic specifications, 
a parameterized specification is defined as a morphism of specifications
where the parameter is the source 
and the parameter passing is defined as a pushout \cite{ADJ80}. 


The framework for this paper is provided by \emph{equational logic}, 
considered from a categorical point of view.  
An equational theory, or simply a theory, is a category with chosen finite products. 
A model $M$ of a theory $\Tt$ is a functor $M\colon \Tt\to\Set$ which maps 
the chosen products to cartesian products. 
A theory $\Tt$ can be presented by a specification $\Ss$,
this means that $\Ss$ generates $\Tt$. 
In this paper, we are not interested in specifications for themselves, 
but as presentations of theories.
So, specifications are used mainly in the examples, 
and we feel free to modify a specification whenever needed 
as long as the presented theory is not changed.


The \emph{parameterization process} studied in this paper 
is essentially the ``\emph{imp} construction'' of \cite{LPR03}.
Starting from a theory $\Tt$ it provides a \emph{parameterized theory}~$\Tt_A$
by adding a \emph{type of parameters} $A$ 
and by transforming each term $f\colon X\to Y$ in $\Tt$ into 
a parameterized term $f'\colon A\times X\to Y$ in $\Tt_A$.
Then clearly $\Tt_A$ generalizes $\Tt$:
the models of $\Tt$ can be identified to the models of $\Tt_A$
which interpret the type of parameters $A$ as a singleton. 
There is another way to relate $\Tt$ and $\Tt_A$,
called the \emph{parameter passing process}, which runs as follows.
By adding to $\Tt_A$ a constant  $a$ (called the \emph{parameter}) of type $A$ we get 
a \emph{theory with parameter}~$\Tt_a$, 
such that for each parameterized term $f'\colon A\times X\to Y$ in $\Tt_A$
there is a term $f'(a,-)\colon X\to Y$ in $\Tt_a$.
Then the \emph{parameter passing morphism} $j\colon \Tt \to \Tt_a$ 
maps each term $f\colon X\to Y$ in $\Tt$ to $f'(a,-)\colon X\to Y$ in $\Tt_a$.
Given a model $M_A$ of $\Tt_A$ 
an \emph{argument} $\alpha$ is an element of the set $M_A(A)$,
it provides a model $M_{A,\alpha}$ of $\Tt_a$ 
which extends $M_A$ and satisfies $M_{A,\alpha}(a)=\alpha$.
Thanks to the parameter passing morphism, 
the model $M_{A,\alpha}$ of $\Tt_a$ gives rise to a model $M$ of $\Tt$ 
such that $M(f)=M_A(f')(\alpha,-)$ for each term $f$ in $\Tt$.
Moreover, under some relevant terminality assumption on $M_A$, 
this correspondence between the arguments $\alpha\in M_A(A)$ 
and the models of $\Tt$ is a bijection: 
this is the \emph{exact parameterization property}.


The parameterization process and its associated parameter passing process
have been described for each given theory $\Tt$,
but in fact they have the property of preserving the theory structure, 
which can be stated precisely in a categorical framework:
this is the aim of this paper. 
The parameterization process is defined as a \emph{functor}:
the construction of the parameterized theory~$\Tt_A$ from the given theory~$\Tt$ is 
a functor, which in addition is left adjoint to the construction of a coKleisli category.
The parameter passing process is defined as a \emph{natural transformation}, 
along the following lines.  
First, the construction of the theory with parameter~$\Tt_a$
from the parameterized theory~$\Tt_A$ is simply a pushout construction,
such that the construction of~$\Tt_a$ from~$\Tt$ is a functor.
Then, each parameter passing morphism $j:\Tt\to\Tt_A$ is defined 
from a lax colimit of theories,
in such a way that the parameter passing morphisms are (essentially) the components 
of a natural transformation from the identity functor to this functor.


A first version of this approach can be found in \cite{DDLR05},
and a more abstract point of view, relying on \emph{diagrammatic logic},
is presented in \cite{eatDiaLog}. 
With respect to the previous papers like \cite{LPR03}, 
we provide a new interpretation of the parameterization process 
and in addition an interpretation of the parameter passing process.
Moreover, we take into account the fact that there is a pure part in the given theory,
and we derive the exact parameterization property
from a more general result which does not rely on the existence of a terminal model.


Equational theories are defined in section~\ref{sec:defi},
then the parameterization process and the parameter passing process
are studied in section~\ref{sec:const}. 
Various examples are presented.
Most of the categorical notions used in this paper can be found in \cite{MacLane98}
or in \cite{BW99}.
We omit the size issues: for instance most colimits should be small. 
A \emph{graph} is a directed multigraph,
and in order to distinguish between various kinds of structures 
with an underlying graph, 
we speak about the \emph{objects} and \emph{morphisms} of a category, 
the \emph{types} and \emph{terms} of a theory or a specification 
and the \emph{points} and \emph{arrows} of a limit sketch.

\section{Definitions} 
\label{sec:defi}

\subsection{Equational theories and specifications} 
\label{subsec:equa}

In this paper, equational logic is seen from a categorical point of view, 
as for instance in \cite{Pitts}. 

\begin{defi}
\label{defi:equa-thry}
The category $\catT_{\eq}$ of \emph{equational theories} 
is made of the categories with chosen finite products
together with the functors which preserve the chosen finite products.
In addition, $\catT_{\eq}$ can be seen as a 2-category 
with the natural transformations as 2-cells. 
\end{defi}

Equational theories are called simply \emph{theories}.
For instance, the theory $\Set$ is made of the category of sets 
with the cartesian products as chosen products.

\begin{rema}
\label{rema:equa-thry}
The correspondence between 
equational theories in the universal algebra style (as in \cite{LEW96}) 
and equational theories in the categorical style (as defined here)  
can be found in \cite{Pitts}. 
Basically, the \emph{sorts} and products of sorts become objects, still called types,
the \emph{operations} and \emph{terms} become morphisms, still called terms 
(the \emph{variables} correspond to projections, as in example~\ref{exam:sgp})
and the \emph{equations} become equalities:
for instance a commutative square $g_1\circ f_1=g_2\circ f_2$ 
means that there is a term $h$ such that $g_1\circ f_1=h$ and $g_2\circ f_2=h$.
A more subtle point of view on equations is presented in \cite{eatLong}.
\end{rema}

\begin{defi}
\label{defi:equa-mod}
A \emph{(strict) model} $M$ of a theory $\Tt$ is 
a morphism of theories $M\colon \Tt\to\Set$
and a \emph{morphism $m\colon M\to M'$ of models} of $\Tt$ is a natural transformation. 
This forms the category $\Mod(\Tt)$ of models of $\Tt$.
\end{defi}

For every morphism of equational theories $\tt\colon \Tt_1\to\Tt$,
we denote by $\tt^\md\colon \Mod(\Tt)\to\Mod(\Tt_1)$ 
the functor which maps each model $M$ of $\Tt$ 
to the model $\tt^\md(M)=M\circ \tt$ of $\Tt_1$
and each morphism $m\colon M\to M'$ to $m\circ \tt$. 
In addition, for each model $M_1$ of $\Tt_1$, 
the category of \emph{models of $\Tt$ extending $M_1$} is denoted $\Mod(\Tt)|_{M_1}$,
it is the subcategory of $\Mod(\Tt)$ 
made of the models $M$ such that $\tt^\md(M)=M_1$
and the morphisms $m$ such that $\tt^\md(m)=\id_{M_1}$.
Whenever $\tt$ is surjective on types, the category $\Mod(\Tt)|_{M_1}$ is discrete.

A theory $\Tt$ can be described by some presentation: 
a \emph{presentation} of an equational theory $\Tt$
is an equational specification $\Ss$ which generates $\Tt$;
this is denoted $\Tt\dashv\Ss$.
Two specifications are called \emph{equivalent} when they present the same theory.
An equational specification can be defined either 
in the universal algebra style as a signature (made of sorts and operations)
together with equational axioms, or equivalently, in a more categorical style,  
as a finite product sketch, see \cite{Lellahi89}, \cite{BW99}.
The correspondence between the universal algebra and the categorical  
points of view runs as in remark~\ref{rema:equa-thry}.

\begin{defi}
\label{defi:equa-spec}
The category $\catS_{\eq}$ of \emph{equational specifications} 
is the category of finite product sketches.
With (generalized) natural transformations as 2-cells, $\catS_{\eq}$ can be seen as a 2-category.
\end{defi}

Equational specifications are called simply \emph{specifications}. 
The category $\catT_{\eq}$ can be identified to a subcategory of $\catS_{\eq}$
(more precisely, to a reflective subcategory of $\catS_{\eq}$).
When $\Ss$ is a presentation of $\Tt$, a model of $\Tt$
is determined by its restriction to $\Ss$, which is called a \emph{model} of $\Ss$,
and in fact $\Mod(\Tt)$
can be identified to the category $\Mod(\Ss)$ of models of $\Ss$.

\begin{figure}[!t]
$$ \begin{array}{|l|l|c|}
\hline
 & \textrm{subscript} \; \ttE & \Ss_{\ttE} \\
\hline
\hline
  \textrm{type (or sort)} & \Type & X \\ 
\hline
  \textrm{term (or operation)} & \Term & 
  \xymatrix@C=3pc{X\ar[r]^{f} & Y} \\ 
\hline
  \textrm{selection of identity} & \Selid &  
  \xymatrix@C=3pc{X\ar[r]^{\id_X} & X} \\
\hline
  \textrm{composition} & \Comp &  
  \xymatrix@C=3pc{X\ar[r]^{f} \ar@/_3ex/[rr]_{g\circ f}^{=} & Y\ar[r]^{g} & Z}  \\ 
\hline
  \textrm{terminal type} & \zProd &   
  \xymatrix{ \uno \\ } \\ 
\hline
  \textrm{collapsing} & \zTuple &  
  \xymatrix@C=3pc{ 
    X \ar[r]^{\tu_X} &\uno  \\ } \\ 
\hline
  \textrm{binary product} & \bProd &   
  \xymatrix@C=3pc@R=.7pc{ X & \\ 
    & X\stimes Y \ar[lu]_{p_X} \ar[ld]^{p_Y} \\ Y & \\ } \\ 
\hline
  \textrm{pairing} & \bTuple &  
  \xymatrix@C=3pc@R=1pc{ & X & \\ 
    Z \ar[ru]^{f} \ar[rd]_{g} \ar[rr]|{\,\tuple{f,g}\,} & 
    \ar@{}[u]|{=}\ar@{}[d]|{=} & 
    X\stimes Y \ar[lu]_{p_X} \ar[ld]^{p_Y} \\ & Y & \\ } \\ 
\hline
\end{array}$$
\caption{\label{fig:elem} Elementary specifications}
\end{figure}

We will repeatedly use the fact that $\catT_{\eq}$ and $\catS_{\eq}$,
as well as other categories of theories and of specifications, have colimits,
and that left adjoint functors preserve colimits. 
In addition every specification is the colimit of a diagram of elementary specifications.
The \emph{elementary specifications} are the specifications respectively made of: 
a type, 
a term, 
an identity term, 
a composed term, 
a $n$-ary product 
and a $n$-ary tuple for all $n\geq0$,
or only for $n=0$ and $n=2$, as in figure~\ref{fig:elem}.
Let us consider a theory $\Tt$ presented by a specification $\Ss$, 
then $\Ss$ is the colimit of a diagram $\Delta$ of elementary specifications,
and $\Tt$ is the colimit of the diagram of theories generated by $\Delta$.

\subsection{Examples} 
\label{subsec:exam}

\begin{exam}
\label{exam:term}

Let us consider the theory $\Tt_{\oper,0}$ presented by two types $X,Y$,
and the three following theories extending $\Tt_{\oper,0}$ 
(the subscript $\oper$ stands for ``operation'',
since $\Tt_{\oper}$ is presented by the elementary specification 
for terms or operations $\Ss_{\Term}$). 
The unit type is denoted $\uno$ and the projections are not given any name.

$$  \begin{array}{c|c|c|c|}
\cline{2-2} 
  \Tt_{\oper,A} \dashv  & 
  \xymatrix@R=1pc{A & A\stimes X\ar[l]\ar[d]\ar[rd]^{f'} & \\ & X & Y \\ } &  
  \multicolumn{2}{c}{} \\
\cline{2-2} \cline{4-4} 
  \multicolumn{2}{c}{} & 
  \qquad \Tt_{\oper,a} \dashv  & 
  \xymatrix@R=1pc{A & A\stimes X\ar[l]\ar[d]\ar[rd]^{f'} & \\ \uno \ar[u]^{a} &X & Y \\ } \\
\cline{2-2} \cline{4-4} 
  \Tt_{\oper} \dashv  & 
  \xymatrix{ & X\ar[r]^{f} &Y \\ } & 
  \multicolumn{2}{c}{} \\
\cline{2-2} 
\end{array}$$

These theories are related by various morphisms (all of them preserving $\Tt_{\oper,0}$):
$\tt_{\oper,A}\colon \Tt_{\oper,A}\to\Tt_{\oper}$ maps $A$ to $\uno$
and $\tt_{\oper,a}\colon \Tt_{\oper,a}\to\Tt_{\oper}$ extends $\tt_{\oper,A}$ 
by mapping $a$ to $\id_{\uno}$,
while $j_{\oper,A}\colon \Tt_{\oper,A}\to\Tt_{\oper,a}$ is the inclusion.
In addition, here are two other presentations of the theory $\Tt_{\oper,a}$
(the projections are omitted and $ \uno\times X$ is identified to $X$):

$$ 
 \begin{array}{|c|}
\cline{1-1}
\xymatrix@R=1pc{A & A\stimes X\ar[rd]^{f'} & \\ 
\uno \ar[u]^{a} & X  \ar@{}[ru]|(.3){=} \ar[u]^{a\stimes\id_X} \ar[r]_{f''} & Y \\ } \\ 
\cline{1-1}
\end{array}  
\qquad\qquad 
\begin{array}{|c|}
\cline{1-1}
\xymatrix@R=1pc{A & A\stimes X\ar[r]^{f'} \ar@{}[rd]|{=} &Y \\ 
\uno \ar[u]^{a} & X \ar[u]^{a\stimes\id_X} \ar[r]_{f''} & Y \ar[u]_{\id_Y}  \\ } \\ 
\cline{1-1}
\end{array} $$
It is clear from these presentations of $\Tt_{\oper,a}$
that there is a morphism $j_{\oper}\colon \Tt_{\oper}\to\Tt_{\oper,a}$ 
which maps $f$ to $f''$.
In addition, $\tt_{\oper,a}\circ j_{\oper,A} = \tt_{\oper,A}$
and there is a natural transformation
$t_{\oper} \colon  j_{\oper} \circ \tt_{\oper,A}  \To j_{\oper,A}$ 
defined by $(t_{\oper})_X=\id_X$, $(t_{\oper})_Y=\id_Y$ 
and $(t_{\oper})_A=a\colon \uno\to A$.

$$ \xymatrix@R=.8pc{
  \Tt_{\oper,A} \ar[dd]_{\tt_{\oper,A}} \\ 
  \mbox{ } \\ 
  \Tt_{\oper}  \\
  }
\qquad \qquad 
\xymatrix@R=.8pc{
  \Tt_{\oper,A} \ar[dr]^{j_{\oper,A}} \ar[dd]_{\tt_{\oper,A}} \\ 
  \ar@{}[r]|(.4){=} & \Tt_{\oper,a} \ar[dl]^{\tt_{\oper,a}} \\ 
  \Tt_{\oper}  & \\
  }
\qquad \qquad 
 \xymatrix@R=.8pc{
  \Tt_{\oper,A} \ar[dr]^{j_{\oper,A}} \ar[dd]_{\tt_{\oper,A}} & \\ 
  \ar@{}[r]|(.4){\trnat}|(.25){t_{\oper}} & \Tt_{\oper,a}  \\ 
  \Tt_{\oper} \ar[ur]_{j_{\oper}} \\ 
  }
$$ 

\textbf{Parameterization process}
(construction of $\Tt_{\oper,A}$ from $\Tt_{\oper}$).
The theory $\Tt_{\oper,A}$ is obtained from $\Tt_{\oper}$
by adding a type $A$, called the \emph{type of parameters},
to the domain of the unique term in $\Tt_{\oper}$.
Then $\Tt_{\oper,A}$ can be seen as a \emph{generalization} of $\Tt_{\oper}$,
since each model $M$ of $\Tt_{\oper}$ can be identified to a model of $\Tt_{\oper,A}$
where $M(A)$ is a singleton. 

\textbf{Parameter passing process}  
(construction of $\Tt_{\oper,a}$ from $\Tt_{\oper,A}$
and of a morphism from $\Tt_{\oper}$ to $\Tt_{\oper,a}$). 
The theory $\Tt_{\oper,a}$ is obtained from $\Tt_{\oper,A}$
by adding a constant term $a\colon \uno\to A$, called the \emph{parameter}.
A model $M_a$ of $\Tt_{\oper,a}$ is made of a model $M_A$ of $\Tt_{\oper,A}$
together with an element $\alpha=M_a(a)\in M_A(A)$,
so that we can denote $M_a=(M_A,\alpha)$.
Now, let $M_A$ be some fixed model of $\Tt_{\oper,A}$,
then the models $M_a$ of $\Tt_{\oper,a}$ extending $M_A$ correspond bijectively 
to the elements of $M_A(A)$ by $M_a \mapsto M_a(a)$,
so that we get the \emph{parameter adding} bijection
(the category $\Mod(\Tt_{\oper,a})|_{M_A}$ is discrete):
  $$ \Mod(\Tt_{\oper,a})|_{M_A} \isoto M_A(A) \hsp \mbox{by} \hsp  
  M_a = (M_A,\alpha) \mapsto M_a(a)=\alpha \;. $$
On the other hand, 
each model $M_a=(M_A,\alpha)$ of $\Tt_{\oper,a}$ gives rise to a model
${j_{\oper}}^\md(M_a)$ of $\Tt_{\oper}$ such that 
${j_{\oper}}^\md(M_a)(X)=M_a(X)=M_A(X)$, ${j_{\oper}}^\md(M_a)(Y)=M_a(Y)=M_A(Y)$
and ${j_{\oper}}^\md(M_a)(f)=M_a(f'')=M_A(f')(\alpha,-) $.
Now, let $M_A$ be some fixed model of $\Tt_{\oper,A}$ 
and $M_0$ its restriction to $\Tt_{\oper,0}$,
then for each model $M_a=(M_A,\alpha)$ of $\Tt_{\oper,a}$ extending $M_A$ 
the model ${j_{\oper}}^\md(M_a)$ of $\Tt_{\oper}$ extends $M_0$.
This yields the \emph{parameter passing} function 
(the categories $\Mod(\Tt_{\oper,a})|_{M_A}$ and $\Mod(\Tt_{\oper})|_{M_0}$ 
are discrete):
  $$ \Mod(\Tt_{\oper,a})|_{M_A} \to \Mod(\Tt_{\oper})|_{M_0}  \hsp \mbox{by} \hsp  
  M_a  \mapsto {j_{\oper}}^\md(M_a) \;. $$

\textbf{Exact parameterization.} 
Let $M_0$ be any fixed model of $\Tt_{\oper,0}$,
it is made of two sets $\setX=M_0(X)$ and $\setY=M_0(Y)$. 
Let $M_A$ be the model of $\Tt_{\oper,A}$ extending $M_0$ such that 
$M_A(A)=\setY^{\setX}$ and 
$M_A(f')\colon \setY^{\setX} \times \setX \to \setY$ 
is the application. 
It can be noted that $M_A$ is the terminal model of $\Tt_{\oper,A}$ extending $M_0$. 
Then the parameter passing function is a bijection, 
and composing it with the parameter adding bijection we get
(where $\curry{M(f)}\in\setY^{\setX}$ corresponds by currying to 
$M(f)\colon \setX\to\setY$):
  $$ \Mod(\Tt_{\oper})|_{M_0} \iso M_A(A) \hsp \mbox{by} \hsp  
  M \leftrightarrow \curry{M(f)} \;. $$ 

\end{exam}

\begin{exam}
\label{exam:sgp}
Let $\Tt_{\sgp}$ be the theory for semigroups 
presented by one type $G$, 
one term $\prd\colon G^2 \to G$ 
and one equation $\prd(x,\prd(y,z))=\prd(\prd(x,y),z)$
where $x$, $y$, $z$ are variables of type $G$.
As usual with the categorical point of view, 
in fact the \emph{variables} are projections; 
here, $x,y,z\colon G^3\to G$ are the three projections 
and $\prd(x,y)$ is $\prd\circ\tuple{x,y}\colon G^3\to G$,
composed of the pair $\tuple{x,y}\colon G^3\to G^2$ and of $\prd\colon G^2 \to G$,
and so on.

\textbf{Parameterization process}. 
In order to get parameterized families of semigroups, 
we consider the theory $\Tt_{\sgp,A}$ 
presented by two types $A$ and $G$, 
one term $\prd'\colon A\times G^2 \to G$ 
and one equation $\prd'(p,x,\prd'(p,y,z))= \prd'(p,\prd'(p,x,y),z)$ 
where $x$, $y$, $z$ are variables of sort $G$ 
and $p$ is a variable of sort $A$. 

\textbf{Parameter passing process}.
The theory $\Tt_{\sgp,a}$ is $\Tt_{\sgp,A}$ together with a parameter $a\colon \uno\to A$,
hence with $\prd''=\prd'\circ(a\times\id_{G^2}) \colon G^2 \to G$ 
(where $\uno \times G^2$ is identified to $G^2$).
Each model $M_A$ of $\Tt_{\sgp,A}$ gives rise to a family of models of $\Tt_{\sgp,a}$,
all of them with the same underlying set $M_A(G)$
but with different interpretations of $a$ in $M_A(A)$.
Mapping $\prd$ to $\prd''$ defines a morphism from $\Tt_{\sgp}$ to $\Tt_{\sgp,a}$.
So, each model $M_a$ of $\Tt_{\sgp,a}$ gives rise to a model $M$ of $\Tt_{\sgp}$ such that 
$M(G)=M_a(G)$ and $M(\prd)(x,y)=M_a(\prd')(\alpha,x,y) $ for each $x,y\in M_a(G)$,
where $\alpha=M_a(a)$ is called the argument.

\end{exam}

\begin{exam}
\label{exam:list}
This example motivates the existence of pure terms in the given theory.
Let us consider the theory $\Tt_{\nat}$ ``of naturals''
presented by a type $N$ and two terms $z\colon \uno\to N$ and $s\colon N\to N$,
and let us say that $z$ is pure. 
Let $\Tt_{\nat,0}$ be the subtheory presented by $N$ and $z$,
it is called the pure subtheory of $\Tt_{\nat}$. 
We define the theory $\Tt_{\nat,A}$ as made of 
two types $A$ and $N$ and two terms 
$z\colon \uno\to N$ and $s'\colon A\times N\to N$.
It should be noted that $\Tt_{\nat,A}$ contains $\eps_\uno\colon A\times \uno\to \uno$ 
and $z'=z\circ \eps_\uno\colon A\times \uno\to N$.
Then $\Tt_{\nat,A}$ is a theory ``of lists of $A$'', with $z$ for the empty list 
and $s'$ for concatenating an element to a list.
In this way, the theory of lists of $A$ is built as a generalization 
of the theory of naturals; 
indeed the naturals can be identified to the lists over a singleton. 
\end{exam}

\begin{exam}
\label{exam:dm}
Here is another example where pure terms are required,
this is a simplified version of many structures in Kenzo/EAT.
Let $\Tt_{\mon}$ be the theory for monoids presented by
one type $G$,
two terms $\prd\colon G^2 \to G$ and $\unt\colon \to G$,
and the equations $\prd(x,\prd(y,z)) = \prd(\prd(x,y),z)$,
$\prd(x,\unt) = x$, $\prd(\unt,x) = x$ where $x$, $y$, $z$ are variables of type $G$.
Let $\Tt_{\dm}$ be the theory for \emph{differential monoids},
presented by $\Tt_{\mon}$ together with one term $\dif \colon G \to G$ 
and the equations 
$\dif(\prd(x,y)) = \prd(\dif(x),\dif(y))$, $\dif(\unt) = \unt$, $\dif(\dif(x)) = \unt$,
and with the terms in $\Tt_{\mon}$ as its pure terms. 
In order to get parameterized families of differential structures on one monoid,
we define the theory $\Tt_{\dm,A}$ 
presented by two types $G$ and $A$, three terms
$\prd\colon G^2 \to G$, $\unt\colon \uno\to G$ and $\dif' \colon A\times G \to G$,
the three equations from $\Tt_{\mon}$ and the equations 
$\dif'(p,(\prd(x,y))) = \prd(\dif'(p,x), \dif'(p,y))$, $\dif'(p, \unt) = \unt$,
$\dif'(p,\dif'(p,x)) = \unt$.
Each model $M_A$ of $\Tt_{\dm,A}$ gives rise to a family
of models of $\Tt_{\dm}$, 
all of them with the same underlying monoid $(M_A(G),M_A(\prd),M_A(\unt))$:
there is a model $M_a$ of $\Tt_{\dm}$ extending $M_A$ 
for each element $\alpha$ in $M_A(A)$,
with its differential structure defined by $M_a(\dif)= M_A(\dif')(\alpha,-)$.
\end{exam}

\begin{exam}
\label{exam:pi}
In the next sections we will use the theories with the following presentations:
$$  \begin{array}{c|c|c|c|}
\cline{2-2} 
\Pi_A \dashv  & \xymatrix{A} &  \multicolumn{2}{c}{} \\
\cline{2-2} \cline{4-4}
\multicolumn{2}{c}{} & \qquad \Pi_a \dashv  & \xymatrix@R=1pc{ A \\ \uno \ar[u]^{a} \\ } \\
\cline{2-2} \cline{4-4}
\Pi \dashv  & \xymatrix{\uno} &  \multicolumn{2}{c}{} \\
\cline{2-2} 
\end{array}$$
These theories are related by several morphisms:
$\pi_A\colon  \Pi_A\to \Pi$ maps $A$ to $\uno$, 
both $i\colon \Pi\to \Pi_a$ and $i_A\colon  \Pi_A\to \Pi_a$ are the inclusions,
and $\pi_a\colon  \Pi_a\to \Pi$ extends $\pi_A$ by mapping $a$ to $\id_{\uno}$,
so that $\pi_A$ and $\pi_a$ are epimorphisms.
In addition, $\pi_a\circ i_A=\pi_A$ and there is a natural transformation
$p\colon  i \circ \pi_A  \To i_A$ defined by $p_A=a\colon \uno\to A$. 
The diagram below on the right is the \emph{lax colimit of $\pi_A$}, 
which means that it enjoys the following universal property:
for each $\Pi'_a$ with $i'_A\colon  \Pi_A\to \Pi'_a$, $i'\colon  \Pi\to \Pi'_a$ 
and $p'\colon i' \circ \pi_A  \To i'_A$,
there is a unique $h\colon  \Pi_a\to \Pi'_a$ such that 
$h\circ i_A = i'_A$, $h\circ i = i'$ and $h\circ p=p'$. 
For instance, given $\Pi$, $\pi_A\colon  \Pi_A\to \Pi$, $\id_{\Pi}\colon  \Pi\to \Pi$ 
and $\id_{\pi_A}\colon \pi_A \To \pi_A$,
then $\pi_a\colon  \Pi_a\to \Pi$ is the unique morphism such that 
$\pi_a\circ i_A = \pi_A$, $\pi_a\circ i = \id_{\Pi}$ and $\pi_a\circ p=\id_{\pi_A}$.
$$ \xymatrix@R=.8pc{
  \Pi_A \ar[dd]_{\pi_A} \\ 
  \\ 
  \Pi  \\
  }
\qquad \qquad 
\xymatrix@R=.8pc{
  \Pi_A \ar[dr]^{i_A} \ar[dd]_{\pi_A} \\ 
  \ar@{}[r]|(.4){=} & \Pi_a \ar[dl]^{\pi_a} \\ 
  \Pi  & \\
  }
\qquad \qquad 
 \xymatrix@R=.8pc{
  \Pi_A \ar[dr]^{i_A} \ar[dd]_{\pi_A} & \\ 
  \ar@{}[r]|(.4){\trnat}|(.25){p} & \Pi_a  \\ 
  \Pi \ar[ur]_{i} \\ 
  }
$$ 
\end{exam}

\subsection{Some other kinds of theories} 
\label{subsec:thry}

For every theory $\Tt$, 
the coslice category of \emph{theories under $\Tt$} is denoted $\Tt\cosl\catT_{\eq}$.
It can be seen as a 2-category, with 
the natural transformations which extend the identity on $\Tt$ as 2-cells.

\begin{defi}
\label{defi:A-th}
A \emph{parameterized theory} $\Tt_A$ 
is a theory $\Tt$ with a distinguished type, 
called the \emph{type of parameters} and usually denoted $A$. 
The 2-category of parameterized theories 
is the coslice 2-category $\catT_A=\Pi_A\cosl\catT_{\eq}$ of theories under $\Pi_A$.
A \emph{theory with a parameter} $\Tt_a$ 
is a parameterized theory with a distinguished constant of type $A$,
called the \emph{parameter} and usually denoted $a\colon \uno\to A$.
The 2-category of theories with a parameter 
is the coslice 2-category $\catT_a=\Pi_a\cosl\catT_{\eq}$ of theories under $\Pi_a$.
\end{defi}

According to the context,
$\Tt_A$ denotes either the parameterized theory $\gamma_A \colon \Pi_A\to\Tt_A$,
or the equational theory $\Tt_A$ itself. 
Similarly for $\Tt_a$, which denotes either $\gamma_a\colon \Pi_a\to\Tt_a$
or $\Tt_a$ itself. 
In addition, it can be noted that $\Pi$ is the initial theory
(which may also be presented by the empty specification) 
so that $\Pi\cosl\catT_{\eq}$ is isomorphic to $\catT_{\eq}$.
The 2-categories $\catS_A$ and $\catS_a$ of \emph{parameterized specifications}
and \emph{specifications with a parameter}, respectively, 
are defined in a similar way.

On the other hand, the input of the parameterization process 
is a theory $\Tt$ together with a wide subtheory $\Tt_0$
(\emph{wide} means: with the same types), 
such a structure is called a decorated theory.

\begin{defi}
\label{defi:dec-th}
A \emph{decorated theory} is made of a theory $\Tt$ with a wide subtheory $\Tt_0$ 
called the \emph{pure} subtheory of $\Tt$.
A morphism of decorated theories is a morphism of theories $\tt\colon \Tt\to\Tt'$ 
which maps the pure part of $\Tt$ to the pure part of $\Tt'$.
This forms the category $\catT_{\dec}$ of decorated theories.  
\end{defi}

So, a decorated theory $\Tt$ is endowed with a distinguished family of terms,
called the \emph{pure} terms, 
such that all the identities and projections are pure  
and every composition or tuple of pure terms is pure.
Pure terms are denoted with ``$\rpto$''.
When there is no ambiguity we often use the same notation $\Tt$ 
for the theory $\Tt$ itself  
and for the decorated theory made of $\Tt$ and $\Tt_0$.
The decorated specifications are defined in a straightforward way.
For instance, we may consider the decorated specification made of 
a type $N$, a pure term $z\colon \uno\rpto N$ and a term $s\colon N\to N$
(see example~\ref{exam:list}).

\section{Constructions} 
\label{sec:const}

\subsection{The parameterization process is a functor} 
\label{subsec:gene}

In this section we prove that the parameterization process is functorial,
by defining a functor $F_{\param}\colon \catT_{\dec}\to\catT_A$,
called the \emph{parameterization functor}, 
which adds the type of parameters to the domain of every non-pure term.
In addition, theorem~\ref{theo:gene} states that $F_{\param}$ is left adjoint to 
the functor $G_{\param}\colon \catT_A\to\catT_{\dec}$, 
which builds the coKleisli category of the comonad $A\times-$.
 

\begin{figure}[!t]
$$ \begin{array}{|l|l|c|c|}
\hline
 & \textrm{index} \; \ttE\deco & \Ss_{\ttE\deco} & F_{\param}\Ss_{\ttE\deco} \\
\hline
\hline
  \textrm{type} & \Type\dotp & X & X \\ 
\hline
  \textrm{pure term} & \Term\dotp & 
  \xymatrix@C=3pc{X\ar@{~>}[r]^{f} & Y} & 
  \xymatrix@C=3pc{X\ar[r]^{f} & Y} \\ 
\hline
  \textrm{term} & \Term\dotg &
  \xymatrix@C=3pc{X\ar[r]^{f} & Y} & 
  \xymatrix@C=3pc{A\stimes X\ar[r]^{f'} & Y}  \\ 
\hline
  \textrm{selection of identity} & \Selid\dotp &
  \xymatrix@C=3pc{X\ar@{~>}[r]^{\id_X} & X} & 
  \xymatrix@C=3pc{X\ar[r]^{\id_X} & X} \\
\hline
  \textrm{pure composition} & \Comp\dotp & 
  \xymatrix@C=3pc{X\ar@{~>}[r]^{f} \ar@/_3ex/@{~>}[rr]_{g\circ f}^{=} & Y\ar@{~>}[r]^{g} & Z} & 
  \xymatrix@C=3pc{X\ar[r]^{f} \ar@/_3ex/[rr]_{g\circ f}^{=} & Y\ar[r]^{g} & Z} \\
\hline
  \textrm{composition} & \Comp\dotg &  
  \xymatrix@C=3pc{X\ar[r]^{f} \ar@/_3ex/[rr]_{g\circ f}^{=} & Y\ar[r]^{g} & Z} & 
  \xymatrix@C=3pc{A\stimes X \ar[r]^{\tuple{\proj_X,f'}} 
     \ar@/_3ex/[rr]_{g'\circ\tuple{\proj_X,f'}}^{=} & A\stimes Y\ar[r]^{g'} & Z \\} \\ 
\hline
  \textrm{terminal type} & \zProd\dotp &  
  \xymatrix{ \uno \\ }  & 
  \xymatrix{ \uno \\ } \\ 
\hline
  \textrm{pure collapsing} & \zTuple\dotp & 
  \xymatrix@C=3pc{ 
    X \ar@{~>}[r]^{\tu_X} & \uno  \\ }  & 
  \xymatrix@C=3pc{ 
    X \ar[r]^{\tu_X} & \uno  \\ }  \\  
\hline
  \textrm{binary product} & \bProd\dotp &  
  \xymatrix@C=3pc@R=.2pc{ X & \\ 
    & X\stimes Y \ar@{~>}[lu]_{p_X} \ar@{~>}[ld]^{p_Y} \\ Y & \\ }  & 
  \xymatrix@C=3pc@R=.2pc{ X & \\ 
    & X\stimes Y \ar[lu]_{p_X} \ar[ld]^{p_Y} \\ Y & \\ } \\ 
\hline
  \textrm{pure pairing} & \bTuple\dotp & 
  \xymatrix@C=3pc@R=.5pc{ & X & \\ 
    Z \ar@{~>}[ru]^{f} \ar@{~>}[rd]_{g} \ar@{~>}[rr]|{\tuple{f,g}} & \ar@{}[u]|{=}\ar@{}[d]|{=} & 
    X\stimes Y \ar@{~>}[lu]_{p_X} \ar@{~>}[ld]^{p_Y} \\ & Y & \\ }  & 
  \xymatrix@C=3pc@R=.5pc{ & X & \\ 
    Z \ar[ru]^{f} \ar[rd]_{g} \ar[rr]|{\tuple{f,g}} & 
    \ar@{}[u]|{=}\ar@{}[d]|{=} & 
    X\stimes Y \ar[lu]_{p_X} \ar[ld]^{p_Y} \\ & Y & \\ } \\ 
\hline
  \textrm{pairing} & \bTuple\dotg & 
  \xymatrix@C=3pc@R=.5pc{ & X & \\ 
    Z \ar[ru]^{f} \ar[rd]_{g} \ar[rr]|{\tuple{f,g}} & \ar@{}[u]|{=}\ar@{}[d]|{=} & 
    X\stimes Y \ar@{~>}[lu]_{p_X} \ar@{~>}[ld]^{p_Y} \\ & Y & \\ }  & 
  \xymatrix@C=3pc@R=.5pc{ & X & \\ 
    A\stimes Z \ar[ru]^{f'} \ar[rd]_{g'} \ar[rr]|{\tuple{f',g'}} & 
    \ar@{}[u]|{=}\ar@{}[d]|{=} & 
    X\stimes Y \ar[lu]_{p_X} \ar[ld]^{p_Y} \\ & Y & \\ } \\ 
\hline
\end{array}$$
\caption{\label{fig:F-elem} The functor $F_{\param}$ 
on elementary decorated specifications}
\end{figure}

In order to define the functor $F_{\param}$ 
we use the fact that it should preserve colimits.
It has been seen in section~\ref{subsec:equa} that 
every specification is the colimit of a diagram of elementary specifications.
Similarly, every decorated specification is the colimit of 
a diagram of elementary decorated specifications, 
denoted $\Ss_{\ttE\deco}$ where 
$\mathtt{x}=\mathtt{p}$ for ``pure'' or $\mathtt{x}=\mathtt{g}$ for ``general''. 
Informally, the functor $F_{\param}$ explicits the fact that every 
general feature  in a decorated specification gets parameterized, 
while every pure feature remains unparameterized.
Figure~\ref{fig:F-elem} defines the parameterized specification 
$F_{\param}(\Ss_{\ttE\deco})$ 
for each elementary decorated specification $\Ss_{\ttE\deco}$
(many projection arrows are omitted, 
when needed the projections from $A\times X$ are denoted 
$\proj_X\colon A\times X\to A$ and $\eps_X\colon A\times X\to X$).
The morphisms of parameterized specifications $F_{\param}(\ss)$,
for $\ss$ between elementary decorated specifications, 
are straightforward. 
For instance, let $\ss_{\ttc}\colon \Ss_{\Term\dotg}\to\Ss_{\Term\dotp}$
be the conversion morphism,
which corresponds to the fact that every pure term can be seen as a general term,
then $F_{\param}(\ss_{\ttc})$ maps 
$f'\colon A\times X\to Y$ in $F_{\param}(\Ss_{\Term\dotg})$ 
to $f\circ \eps_X\colon A\times X\to Y$ in $F_{\param}(\Ss_{\Term\dotp})$.
Now, given a decorated theory $\Tt$ presented by 
the colimit of a diagram $\Delta$ of elementary decorated specifications,
we define $F_{\param}(\Tt)$ as the parameterized theory presented by 
the colimit of the diagram $F_{\param}(\Delta)$ of parameterized specifications.

\begin{defi}
\label{defi:gene}
The functor $F_{\param}:\catT_{\dec}\to\catT_A$ defined above 
is called the \emph{parameterization functor}.
\end{defi}

Clearly the parameterization functor preserves colimits. 
In addition, let $\Tt_A$ be the parameterized theory $F_{\param}(\Tt)$,
it follows from the definition of $F_{\param}$ that the equational theory $\Tt_A$
is a theory under $\Tt_0$. 


Now the functor $G_{\param}$ is defined independently from $F_{\param}$.
Let $\Tt_A$ be a parameterized theory.
The endofunctor of product with $A$ forms a comonad on $\Tt_A$
with the counit $\eps$ made of the projections $\eps_X\colon A\times X\to X$
and the comultiplication made of the terms 
$\delta_X\colon A\times X\to A\times A\times X$ induced by the diagonal on $A$.
Let $\Tt$ be the coKleisli category of this comonad:
it has the same types as $\Tt_A$
and a term $\kl{f}\colon X\to Y$ for each term $f\colon A\times X\to Y$ in $\Tt_A$. 
There is a functor from $\Tt_A$ to $\Tt$ which is the identity on types 
and maps every $g\colon X\to Y$ in $\Tt_A$ 
to $\kl{g\circ \eps_X}\colon X\to Y$ in $\Tt$. 
Then every finite product in $\Tt_A$ is mapped to a finite product in $\Tt$,
which makes $\Tt$ a theory. 
Let $\Tt_0$ denote the image of $\Tt_A$ in $\Tt$, it is a wide subtheory of $\Tt$.
In this way, any parameterized theory yields a decorated theory.
The definition of $G_{\param}$ on morphisms is straightforward,
and the next result follows easily.

\begin{theo}
\label{theo:gene}
The parameterization functor $F_{\param}$ and the functor $G_{\param}$ 
form an adjunction $F_{\param}\dashv G_{\param}$:
  $$ \xymatrix@C=4pc{ \catT_{\dec} \ar@/^/[r]^{F_{\param}} \ar@{}[r]|{\bot} & 
  \catT_A \ar@/^/[l]^{G_{\param}} } $$
\end{theo}

The next result states that $\Tt$ can be easily recovered from $\Tt_A$, 
by mapping $A$ to $\uno$.

\begin{prop}
\label{prop:general}
Let $\Tt$ be a decorated theory with pure subtheory $\Tt_0$
and $\gamma_A\colon \Pi_A\to\Tt_A$ the parameterized theory $F_{\param}(\Tt)$.
Let $\gamma\colon \Pi\to\Tt$ be the unique morphism from the initial theory $\Pi$ 
to the theory $\Tt$.
Then there is a morphism $\tt_A\colon \Tt_A\to\Tt$ under $\Tt_0$
such that the following square is a pushout:
$$ \xymatrix{ 
  \ar@{}[rd]|{\PO} 
  \Pi_A \ar[d]_{\pi_A} \ar[r]^{\gamma_A} & \Tt_A \ar[d]^{\tt_A} \\ 
  \Pi \ar[r]^{\gamma} &\Tt \\  
  } $$
\end{prop}

\begin{proof}
It can easily be checked that this property is satisfied by each elementary specification.
Then the result follows by commuting two colimits:
on the one hand the colimit that defines the given theory from its elementary components,
and on the other hand the pushout.
\end{proof}

When there is an epimorphism of theories $\tt\colon \Tt_1\to\Tt_2$,
we say that $\Tt_1$ is \emph{the generalization of $\Tt_2$ along $\tt$}.
Indeed, since $\tt$ is an epimorphism, 
the functor $\tt^\md\colon \Mod(\Tt_2)\to\Mod(\Tt_1)$ is a monomorphism,
which can be used for identifying $\Mod(\Tt_2)$ to a subcategory of $\Mod(\Tt_1)$.

\begin{coro}
\label{coro:general}
With notations as in proposition~\ref{prop:general},
$\Tt_A$ is the generalization of $\Tt$ along $\tt_A$.
\end{coro}

\begin{proof}
Clearly $\pi_A\colon \Pi_A\to\Pi$ is an epimorphism.
Since epimorphisms are stable under pushouts,
proposition~\ref{prop:general} proves that $\tt_A\colon \Tt_A\to\Tt$ 
is also an epimorphism.
\end{proof}

Let $F_{\param}:\catT_{\dec}\to\catT_A$ be the parameterization functor
and let $U\colon \catT_A\to\catT_{\eq}$ be the functor 
which simply forgets that the type $A$ is distinguished,
so that $U \circ F_{\param}\colon \catT_{\dec}\to\catT_{\eq}$ 
maps the decorated theory $\Tt$ to the equational theory $\Tt_A$.
 $$ \xymatrix@C=4pc{ \catT_{\dec} \ar[r]^{F_{\param}} & 
  \catT_A  \ar[r]^{U} &  \catT_{\eq} \\ } $$
Every theory $\Tt$ can be seen as a decorated theory where 
the pure terms are defined inductively as the identities, the projections, 
and the compositions and tuples of pure terms.
Let $I\colon \catT_{\eq}\to\catT_{\dec}$ denote the corresponding inclusion functor.
Then the endofunctor $U \circ F_{\param}\circ I\colon \catT_{\eq}\to\catT_{\eq}$ 
corresponds to the ``\emph{imp} construction'' of \cite{LPR03},
which transforms each term $f\colon X\to Y$ in $\Tt$ into 
$f'\colon A\times X\to Y$ for a new type $A$.

\subsection{The parameter passing process is a natural transformation}
\label{subsec:passing}


A theory $\Tt_a$ with a parameter is built simply by adding 
a constant $a$ of type $A$ to a parameterized theory $\Tt_A$. 
Obviously, this can be seen as a pushout. 

\begin{defi}
\label{defi:adding}
Let $\gamma_A\colon \Pi_A\to\Tt_A$ be a parameterized theory.
The theory with parameter \emph{extending} $\gamma_A$ 
is $\gamma_a\colon \Pi_a\to\Tt_a$ given by the pushout of $\gamma_A$ and $i_A$:
$$ \xymatrix{ 
  \ar@{}[rd]|{\PO} 
  \Pi_A \ar[d]_{i_A} \ar[r]^{\gamma_A} & \Tt_A \ar[d]^{j_A} \\ 
  \Pi_a \ar[r]^{\gamma_a} &\Tt_a \\  
  } $$
\end{defi}

This pushout of theories gives rise to a pullback of categories of models, 
hence for each model $M_A$ of $\Tt_A$ 
the function which maps each model $M_a$ of $\Tt_a$ extending $M_A$
to the element $M_a(a)\in M_A(A)$ defines a bijection:
\begin{equation}
\label{eq:adding}
 \Mod(\Tt_a)|_{M_A}  \isoto M_A(A) \;. 
\end{equation}

Let us assume that the parameterized theory $\gamma_A\colon \Pi_A\to\Tt_A$ 
is $F_{\param}(\Tt)$ for some decorated theory $\Tt$
with pure subtheory $\Tt_0$.
Then the pushout property in definition~\ref{defi:adding} 
ensures the existence of a unique $\tt_a\colon  \Tt_a\to \Tt$
such that $\tt_a \circ \gamma_a = \gamma\circ\pi_a$
and $\tt_a \circ j_A = \tt_A$,
which means that $\tt_a$ maps $A$ to $\uno$ and $a$ to $\id_{\uno}$
and the $\tt_a$ extends $\tt_A$.
Then $\Tt_A$ is a theory under $\Tt_0$ 
and the composition by $j_A$ makes $\Tt_a$ a theory under $\Tt_0$ 
with $j_A$ preserving~$\Tt_0$.

$$ 
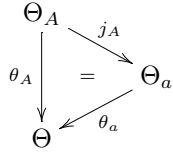
\xymatrix@R=.8pc{
  \Tt_A \ar[dr]^{j_A} \ar[dd]_{\tt_A} & \\
  \ar@{}[r]|(.4){=} & \Tt_a \ar[dl]^{\tt_a} \\ 
  \Tt & \\ 
} $$


\begin{defi}
\label{defi:lax}
For each decorated theory $\Tt$ with pure subcategory $\Tt_0$, 
let $\Tt_A=F_{\param}(\Tt)$ and $\tt_A\colon \Tt_A\to\Tt$ 
as in proposition~\ref{prop:general}, 
and let $\Tt_a$ and $j_A\colon \Tt_A\to\Tt_a$ 
as in definition~\ref{defi:adding}.
Then $j\colon \Tt\to\Tt_a$ is the morphism under $\Tt_0$ 
which maps each type $X$ to $X$ 
and each term $f\colon X\to Y$ to $f'\circ(a\times\id_X)\colon X\to Y$.
And $t\colon j\circ\tt_A \To j_A$ is the natural transformation under $\Tt_0$ 
such that $t_A=a\colon \uno\to A$.
\end{defi}

Lax cocones and lax colimits in 2-categories generalize 
cocones and colimits in categories,
so that the following diagram is a \emph{lax cocone with base $\tt_A$} 
in the 2-category $\Tt_0\cosl\catT_{\eq}$, for short it is denoted $(\Tt_a,j_A,j,t)$,
and it is called \emph{the lax colimit associated to}~$\Tt$ because of lemma~\ref{lemm:passing}.
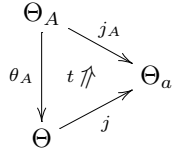
$$ \xymatrix@R=.8pc{
  \Tt_A \ar[dr]^{j_A} \ar[dd]_{\tt_A} & \\
  \ar@{}[r]|(.4){\trnat}|(.25){t}  & \Tt_a  \\ 
  \Tt \ar[ur]_{j} & \\ 
} $$

\begin{lemm}
\label{lemm:passing}
Let $\Tt$  be a decorated theory with pure subcategory $\Tt_0$.
The lax cocone $(\Tt_a,j_A,j,t)$ with base $\tt_A$ defined above 
is a lax colimit in the 2-category of theories under $\Tt_0$.
\end{lemm}

\begin{proof}
This means that the given lax cocone 
is initial among the lax cocones with base $\tt_A$ in $\Tt_0\cosl\Tt$,
in the following sense:
for every lax cocone $(\Tt'_a,j'_A,j',t')$ with base $\tt_A$ under $\Tt_0$
there is a unique morphism $ h\colon \Tt_a\to\Tt'_a$ such that 
$ h\circ j_A=j'_A$, $ h\circ j=j'$ and $ h\circ t = t'$.
Indeed, $h$ is defined from the pushout in definition~\ref{defi:adding} 
by $h\circ j_A=j'_A$, so that $h(A)=A$, 
and $h\circ \gamma_a(a)=t'_A:\uno\to A$. 
\end{proof}

For instance, given $\Tt$, $\tt_A\colon  \Tt_A\to \Tt$, $\id_{\Tt}\colon  \Tt\to \Tt$ 
and $\id_{\tt_A}\colon \tt_A \To \tt_A$,
then $\tt_a$ is the unique morphism such that 
$\tt_a\circ j_A = \tt_A$, $\tt_a\circ j = \id_{\Tt}$ and $\tt_a\circ t=\id_{\tt_A}$.

Let $\Tt$ be a decorated theory with pure subtheory $\Tt_0$
and let $(\Tt_a,j_A,j,t)$ be its associated lax colimit, with base $\tt_A\colon \Tt_A\to \Tt$.
Let $M_A$ be a model of $\Tt_A$
and $M_0$ its restriction to $\Tt_0$, 
and let $\{ (M,m) \mid  m\colon  {\tt_A}^\md M \to M_A  \}|_{M_0}$
(where as before ${\tt_A}^\md M=M\circ\tt_A$) 
denote the set of pairs $(M,m)$ with $M$ a model of $\Tt$ extending $M_0$
and $m$ a morphism of models of $\Tt_A$ extending $\id_{M_0}$. 
A consequence of the lax colimit property 
is that the function which maps each model $M_a$ of $\Tt_a$ extending $M_A$
to the pair $(j^\md M_a, t^\md M_a)=(M_a\circ j,M_a\circ t)$ 
defines a bijection: 
\begin{equation}
\label{eq:passing}
 \Mod(\Tt_a)|_{M_A}  \iso
 \{ (M,m) \mid  m\colon  {\tt_A}^\md M \to M_A  \}|_{M_0} \;.
\end{equation}

The bijections~\ref{eq:adding} and~\ref{eq:passing} provide the next result,
which does not involve $\Tt_a$. 

\begin{prop} 
\label{prop:passing}
Let $\Tt$ be a decorated theory with pure subtheory $\Tt_0$
and let $\Tt_A=F_{\param}(\Tt)$ and $\tt_A\colon \Tt_A\to \Tt$.
Then for each model $M_A$ of $\Tt_A$, 
with $M_0$ denoting the restriction of $M_A$ to $\Tt_0$, 
the function which maps each element $\alpha\in M_A(A)$
to the pair $(M,m)$, 
where $M$ is the model of $\Tt$ such that $M(f)=M_A(f')(\alpha,-)$ 
and where $m:{\tt_A}^\md M \to M_A$ is the morphism of models of $\Tt_A$ 
such that $m_A:M(\uno)\to M_A(A)$ is the constant function $\alpha$, 
defines a bijection: 
\begin{equation}
\label{eq:adding-passing}
M_A(A)  \iso
 \{ (M,m) \mid  m\colon  {\tt_A}^\md M \to M_A  \}|_{M_0} \;.
\end{equation}
\end{prop}

As an immediate consequence, 
we get the \emph{exact parameterization} property from \cite{LPR03}. 

\begin{coro} 
\label{coro:exact}
Let $\Tt$ be a decorated theory with pure subcategory $\Tt_0$,
and let $\Tt_A=F_{\param}(\Tt)$.
Let $M_0$ be a model  of $\Tt_0$
and $M_A$ a terminal model of $\Tt_A$ extending $M_0$. 
Then there is a bijection: 
\begin{equation}
\label{eq:exact}
M_A(A)  \iso \Mod(\Tt)|_{M_0} 
\end{equation}
which maps each $\alpha\in M_A(A)$ to the model 
$M_{A,\alpha}$ of $\Tt$ defined by 
$M_{A,\alpha}(X)=M_0(X)$ for each type $X$
and $M_{A,\alpha}(f)=M_A(f')(\alpha,-)$ for each term $f$,
so that $M_{A,\alpha}(f)=M_A(f)$ for each pure term $f$.
\end{coro} 

The existence of a terminal model of $\Tt_A$ extending $M_0$ 
is a consequence of \cite{Ru00} and \cite{HR95}.
Corollary~\ref{coro:exact} corresponds to the way algebraic structures are implemented 
in the systems Kenzo/EAT.
In these systems 
the parameter set is encoded by means of a record of Common Lisp functions, 
which has a field for each operation in the algebraic structure to be implemented. 
The pure terms correspond to functions which can be obtained 
from the fixed data and do not require an explicit storage. 
Then, each particular instance of the record gives rise to an algebraic structure.


Clearly the construction of $\gamma_a$ from $\gamma_A$ 
is a functor, which is left adjoint to the functor 
which simply forgets that the constant $a$ is distinguished.
So, by composing this adjunction 
with the adjunction $F_{\param}\dashv G_{\param}$ from theorem~\ref{theo:gene} 
we get an adjunction $F'_{\param}\dashv G'_{\param}$
where $F'_{\param}$ maps each decorated theory $\Tt$ to $\Tt_a$, as defined above: 
  $$ \xymatrix@C=4pc{ \catT_{\dec} \ar@/^/[r]^{F'_{\param}} \ar@{}[r]|{\bot} & 
  \catT_a \ar@/^/[l]^{G'_{\param}} } $$
Let $U'\colon \catT_a\to\catT_{\eq}$ be the functor 
which simply forgets that the type $A$ and the constant $a$ are distinguished. 
Then the functor $U' \circ F'_{\param} \colon \catT_{\dec}\to\catT_{\eq}$ 
maps the decorated theory $\Tt$ to the equational theory $\Tt_a$.
  $$ \xymatrix@C=3pc{ 
  \catT_{\dec} \ar[r]^{F'_{\param}} &  \catT_a  \ar[r]^{U'} &  \catT_{\eq} \\ } $$ 
The morphism of theories $j\colon \Tt\to\Tt_a$ from definition~\ref{defi:lax}
depends on the decorated theory $\Tt$, let us denote it $j=J_{\Tt}$.
Let $H\colon \catT_{\dec}\to\catT_{\eq}$ be the functor which maps 
each decorated theory $\Tt$ to the equational theory $\Tt$.
The next result is easy to check.

\begin{theo}
\label{theo:passing}
The morphisms of theories $J_{\Tt}\colon \Tt\to\Tt_a$ form the components of 
a natural transformation $J\colon H \To U'\circ F'_{\param}\colon \catT_{\dec}\to\catT_{\eq}$.
  $$ \xymatrix@C=3pc{ 
  \catT_{\dec} \ar[r]^{F'_{\param}}  \ar@/_4ex/[rr]_{H}^(.4){\Uparrow}^(.45){J}  & 
  \catT_a  \ar[r]^{U'} &  \catT_{\eq} \\ } $$ 
\end{theo}

\begin{defi}
\label{defi:passing}
The natural transformation $J\colon H \To U'\circ F'_{\param}\colon \catT_{\dec}\to\catT_{\eq}$
in theorem~\ref{theo:passing} is called the \emph{parameter passing natural transformation}.
\end{defi}

\subsection{Examples} 
\label{subsec:more-exam}

\begin{exam}
\label{exam:adding}
Starting from $\Tt_{\oper}$ and $\Tt_{\oper,0}$ as in example~\ref{exam:term}, 
the pushouts of theories from proposition~\ref{prop:general} 
and definition~\ref{defi:adding} are respectively: 

$$ \begin{array}{|c|c|c|c|c|c|c|}
\cline{1-1} \cline{3-3} \cline{5-5} \cline{7-7} 
  \xymatrix{A \\} &  
  \longrightarrow & 
  \xymatrix@R=1pc{A & A\stimes X \ar[l]\ar[d]\ar[rd]^{f'} & \\ & X & Y \\} &
  \qquad  \qquad & 
  \xymatrix@R=1pc{A \\} &  
  \longrightarrow & 
  \xymatrix@R=1pc{A & A\stimes X \ar[l]\ar[d]\ar[rd]^{f'} & \\ & X & Y \\} \\
\cline{1-1} \cline{3-3} \cline{5-5} \cline{7-7} 
  \col{\downarrow} & 
  \col{} & 
  \col{\downarrow} & 
  \col{} & 
  \col{\downarrow} & 
  \col{} & 
  \col{\downarrow} \\
\cline{1-1} \cline{3-3} \cline{5-5} \cline{7-7} 
  \xymatrix@R=1pc{ \\ \uno  \\ } &  
  \longrightarrow & 
  \xymatrix@R=1pc{ \\ & X\ar[r]^{f} &Y \\ } & & 
  \xymatrix@R=1pc{A \\ \uno \ar[u]^{a} \\ } &  
  \longrightarrow & 
  \xymatrix@R=1pc{ A & A\stimes X \ar[l]\ar[d]\ar[rd]^{f'} & \\ \uno \ar[u]^{a} & X & Y \\ } \\ 
\cline{1-1} \cline{3-3} \cline{5-5} \cline{7-7} 
\end{array}  $$

We have seen in example~\ref{exam:term} 
two other presentations of the vertex $\Tt_{\oper,a}$ of the second pushout,
with $f''=f'\circ(a\times\id_X):X\to Y$.
For each decorated theory $\Tt$, 
the morphism of equational theories $j_{\oper}=J_{\Tt_{\oper}}:\Tt\to\Tt_a$ 
maps $f$ to $f''$, as in example~\ref{exam:term}.

A model $M_0$ of $\Tt_{\oper,0}$ is simply made of two sets 
$\setX=M_0(X)$ and $\setY=M_0(Y)$.
On the one hand, a model of $\Tt$ extending $M_0$ is characterized 
by a function $\varphi\colon \setX\to\setY$.
On the other hand, 
the terminal model $M_A$ of $\Tt_{\oper,A}$ extending $M_0$ 
is such that $M_A(A)=\setY^\setX$ and 
$M_A(f')\colon \setY^{\setX} \times \setX \to \setY$ is the application. 
The bijection $\Mod(\Tt)|_{M_0}\iso M_A(A)$ then 
corresponds to the currying bijection $\varphi\mapsto\curry{\varphi}$. 
\end{exam}

\begin{exam}
\label{exam:dm-terminal}
Let $\Tt_{\dm}$ be the theory for differential monoids from example~\ref{exam:dm},
with the pure subtheory $\Tt_{\dm,0}=\Tt_{\mon}$ of monoids.
They generate the parameterized theory $\Tt_{\dm,A}$ as in example~\ref{exam:dm}.
Let $M_0$ be some fixed monoid 
and $M_A$ any model of $\Tt_{\dm,A}$ extending $M_0$,
then each element of $M_A(A)$ corresponds to a
differential structure on the monoid $M_0$.
If in addition $M_A$ is the terminal model of $\Tt_{\dm,A}$ extending $M_0$, 
then this correspondence is bijective.
\end{exam}

\begin{exam}
\label{exam:state}
When dealing with an imperative language, 
the states for the memory are endowed with an operation $\lup$
for observing the state and an operation $\upd$ for modifying it. 
There are two points of view on this situation: either the state is hidden,
or it is explicit.
Let us check that the parameterization process allows to generate 
the theory with explicit state from the theory with hidden state.

First, let us focus on observation: 
the theory $\Tt_{\st}$ is made of two types $L$ and $Z$
(for locations and integers, respectively) and a term $v\colon L\to Z$
for observing the values of the variables.
The pure subtheory $\Tt_{\st,0}$ is made of $L$ and $Z$.
We choose a model $M_0$ of $\Tt_{\st,0}$ made of a countable set 
of locations (or adresses, or ``variables'') $\bL=M_0(L)$ 
and of the set of integers $\bZ=M_0(Z)$.
Let $\bA=\bZ^\bL$, then as in example~\ref{exam:adding} 
the terminal model $M_A$ of $\Tt_{\st,A}$ extending $M_0$ 
is such that $M_A(A)=\bA$ and $M_{\st,A}(v')\colon \bA\times\bL\to\bZ$ 
is the application, denoted $\lup$. 
The terminal model $M_A$ does correspond to an ``optimal'' implementation 
of the state.

Now, let us look at another model $N_A$ of $\Tt_{\st,A}$ extending $M_0$,
defined as follows: 
$N_A(A)=\bA\times \bL\times \bZ$ and 
$N_A(v')\colon \bA\times \bL\times \bZ\times \bL\to\bZ$
maps $(p,x,n,y)$ to $n$ if $x=y$ and to $\lup(p,y)$ otherwise.
The terminality property of $M_A$ ensures that there is a unique 
function $\upd\colon \bA\times \bL\times \bZ \to \bA$ such that 
$\lup(\upd(p,x,n),y)$ is $n$ if $x=y$ and is $\lup(p,y)$ otherwise.
So, the operation $\upd$ is defined coinductively from the operation $\lup$. 
\end{exam}

\section{Conclusion}

This paper provides a neat categorical formalization for 
the parameterization process in Kenzo and EAT.
An additional level of abstraction allows to see 
the parameterization process as a morphism of logics and the 
parameter passing process as a 2-morphism of logics,
in a relevant 2-category of logics \cite{eatDiaLog}. 
Future work includes the generalization of this approach 
from equational theories to other families of theories, like distributive categories,
and to more general kinds of parameters, like data types.


\end{document}